%% file: main.tex
\documentclass[runningheads]{llncs}
\usepackage[T1]{fontenc}
\usepackage{graphicx}
\graphicspath{{./graphics/}}

\usepackage[ruled,linesnumbered]{algorithm2e}
\SetKwComment{Comment}{/* }{ */}

\usepackage{tikz}
\usepackage{amsmath}
\usepackage{amssymb}
\usepackage{adjustbox}
\usepackage{caption}
\usepackage{cleveref}
\usepackage{todonotes}
\usepackage{lipsum}
\usepackage{mathtools}
\usetikzlibrary{matrix,decorations.pathreplacing, calc, positioning}

\input{makros}

\begin{document}
\title{Longest Common Subsequence with Gap Constraints}%\thanks{Supported by organization x.}}
\author{Duncan Adamson \and % \inst{1} \and
Maria Kosche\and %\inst{1}\and %\orcidID{0000-0002-2165-2695} \and
Tore Ko\ss \and %\inst{1}\and %\orcidID{0000-0001-6002-1581} \and
Florin Manea \and %\inst{1}\and %\orcidID{0000-0001-6094-3324} \and
Stefan Siemer %\inst{1} %\orcidID{0000-0001-7509-8135}}
}

\institute{Department of Computer Science, University of Göttingen, Göttingen, Germany}
% \institute{\email{duncan.adamson@cs.uni-goettingen.de}}

%
\authorrunning{D. Adamson et al.}

\maketitle              % typeset the header of the contribution

\input{abstract}

\input{introduction}

\input{preliminaries}

\input{specc}

\input{onec}

\input{alphc}

\input{bwin}

\input{conclusion}

\bibliographystyle{splncs04}
\bibliography{bibliography}

%\appendix

%\input{appendix}

\end{document}

%% file: makros.tex
\DeclareMathOperator{\subseq}{\preceq}

\DeclareMathOperator{\polylog}{polylog}

\DeclareMathOperator{\bigO}{O}

\DeclareMathOperator{\emptyword}{\varepsilon}

\newcommand{\alphabet}[1]{\textsf{alph}(#1)}

\newcommand{\gaptuple}{gc}
\newcommand{\gap}[3]{\mathsf{gap}_{#2}(#1, #3)}

\newcommand{\subseqSet}[2]{SubSeq(#1, #2)}
\newcommand{\comSubseqSet}[3]{ComSubSeq(#1, #2, #3)}

\newcommand{\len}[1]{\lvert #1 \rvert}

\newcommand{\Problem}[1]{\textsc{#1}}

\newcommand{\set}[1]{\lbrace #1 \rbrace}

%% file: abstract.tex
\begin{abstract}
We consider the longest common subsequence problem in the context of subsequences with gap constraints. In particular, following Day et al. 2022, we consider the setting when the distance (i.\,e., the gap) between two consecutive symbols of the subsequence has to be between a lower and an upper bound (which may depend on the position of those symbols in the subsequence or on the symbols bordering the gap) as well as the case where the entire subsequence is found in a bounded range (defined by a single upper bound), considered by Kosche et al. 2022. In all these cases, we present efficient algorithms for determining the length of the longest common constrained subsequence between two given strings.
%\keywords{First keyword  \and Second keyword \and Another keyword.}
\end{abstract}

%% file: introduction.tex
\section{Introduction}
\vspace*{-0.2cm}

A {\em subsequence} of a string $w = w[1] w[2] \ldots w[n]$, where $w[i]$ is a symbol from a finite alphabet $\Sigma$ for $i\in \{1,\ldots,n\}$, is a string $v = w[i_1] w[i_2] \ldots w[i_k]$, with $k \leq n$ and $1 \leq i_1 < i_2 \leq \ldots < i_{k} \leq n$. The positions $i_1,i_2,\ldots,i_k$ on which the symbols of $v$ appear in $w$ are said to define the embedding of $v$ in $w$. 

In general, the concept of subsequences appears and plays important roles in many different areas of theoretical computer science such as: formal languages and logics (e.\,g., in connection to piecewise testable languages~\cite{simonPhD,Simon72,KarandikarKS15,CSLKarandikarS,journals/lmcs/KarandikarS19}, or in connection to subword-order and downward-closures~\cite{HalfonSZ17,KuskeZ19,Kuske20,Zetzsche16}); combinatorics on words (e.\,g., in connection to binomial equivalence, binomial complexity, or to subword histories ~\cite{RigoS15,FreydenbergerGK15,LeroyRS17a,Rigo19,Seki12,Mat04,Salomaa05}); the design and complexity of algorithms. To this end, we mention some classical algorithmic problems such as the computation of {\em longest common subsequences} or of the {\em shortest common supersequences} \cite{chvatal,Hirschberg77,HuntS77,Maier:1978,MasekP80,NakatsuKY82,DBLP:journals/tcs/Baeza-Yates91,BergrothHR00}, the testing of the Simon congruence of strings and the computation of the arch-factorisation and universality of words \cite{TCS::Hebrard1991,garelCPM,SimonWords,DBLP:conf/wia/Tronicek02,CrochemoreMT03,dlt2020,DayFKKMS21,KufMFCS,GawrychowskiEtAl2021,KoscheKMS21}; see also \cite{SurveyNCMA} for a survey of some combinatorial algorithmic problems related to subsequence matching. Moreover, these problems and some other closely related ones have recently regained interest in the context of fine-grained complexity (see~\cite{DBLP:conf/fsttcs/BringmannC18,BringmannK18,AbboudEtAl2015,AbboudEtAl2014}). Nevertheless, subsequences appear also in more applied settings: for modelling concurrency~\cite{Riddle1979a,Shaw1978,BussSoltys2014}, in database theory (especially \emph{event stream processing}~\cite{ArtikisEtAl2017,GiatrakosEtAl2020,ZhangEtAl2014}), in data mining~\cite{LiW08,LiYWL12}, or in problems related to bioinformatics~\cite{BilleEtAl2012}. 

Most problems related to subsequences are usually considered in the setting where the embedding of subsequences in words are arbitrary. However, in \cite{DayKMS22}, a novel setting is considered, based on the intuition that, in practical scenarios, some properties with respect to the \emph{gaps} that are induced by the embeddings can be inferred. As such, \cite{DayKMS22} introduces the notion of \emph{subsequences with gap constraints}: these are strings $v$ which can be embedded by some mapping $e$ in a word $w$ in such a way that the \emph{gaps} of the embedding, i.\,e., the factors between the images of the mapping, satisfy certain properties. The main motivation of introducing and studying this model of subsequences in \cite{DayKMS22} comes from data-base theory \cite{KleestMeissnerEtAl2021,Kleest-Meissner23}, and the properties which have to be satisfied by the gaps are specified either in the form of length constraints (i.\,e., bounds on the length of the gap) or regular-language constraints. We refer the reader to \cite{DayKMS22} for a detailed presentation of subsequences with gap constraints and their motivations, as well as a discussion of the various related models. The main results of \cite{DayKMS22} are related to the complexity of the matching problem: given two strings $w$ and $v$, decide if there is an embedding $e$ of $v$ as a subsequence of $w$, such that the gaps induced by $e$ fulfil some given length and regular gap constraints. A series of other complexity results related to analysis problems for the set of subsequences of a word, which fulfil a given set of gap constraints, were obtained. The results of \cite{DayKMS22} are further extended in \cite{KoscheKMP22}, where the authors consider \emph{subsequences in bounded ranges}: these are strings $v$ which can be embedded by some mapping $e$ in a word $w$ in such a way that the range in which all the symbols of $v$ are embedded has length at most $B$, for some given integer $B$. This investigation was motivated in the context of sliding window algorithms \cite{GanardiHKLM18,GanardiHL16,GanardiHL18,GanardiHL18b,GanardiHLS19}, and the obtained results are again related to the complexity of matching and analysis problems.

One of the most studied algorithmic problem for subsequences is the problem of finding the length of the longest common subsequence of two strings (for short \Problem{LCS}), see, e.\,g.,  \cite{chvatal,Hirschberg77,HuntS77,Maier:1978,MasekP80,NakatsuKY82,DBLP:journals/tcs/Baeza-Yates91} or the survey \cite{BergrothHR00}. In this problem, we are given two strings $v$ and $w$, of length $m$ and $n$, respectively, over an alphabet of size $\sigma$, and want to find the largest $k$ for which there exists a string of length $k$ which can be embedded as a subsequence in both $v$ and $w$. The results on \Problem{LCS} are efficient algorithms (in most of the papers cited above) but also conditional complexity lower bounds~\cite{AbboudEtAl2014,AbboudEtAl2015,AbboudRubinstein2018}. In particular, there is a folklore algorithm solving \Problem{LCS} in $O(N)$ time, for $N=mn$, and, interestingly, the existence of an algorithm whose complexity is $O(N^{1-c})$, for some $c>0$, would refute the Strong Exponential Time Hypothesis (SETH), see~\cite{AbboudEtAl2015}. 

\paragraph{Our Contributions.} In this paper, we investigate the \Problem{LCS} problem in the context of subsequences with gap-length constraints, which seem to have a strong motivation and many application (see \cite{DayKMS22,KoscheKMP22} and the references therein). Clearly, in the model considered by \cite{DayKMS22}, the gap constraints depend on the length of the subsequence, while in \Problem{LCS} this is not known, and we actually need to compute this length. So, the model of \cite{DayKMS22} needs to be adapted to the setting of {\Problem{LCS}}. One way to do this is to consider that all the gaps of the common subsequence we search for are restricted by the same pair of lower and upper bounds; in this case, the length of the common subsequence plays no role anymore. We extend this initial idea significantly. On the one hand, we consider the case when there is a constant number of different length constraints which restrict the gaps (and they are given alongside the words). A further extension is the case when we are given, alongside the input words, an arbitrarily long tuple of gap-length constraints: the $i^{th}$ gap constraint in this tuple refers to the gap between the $i^{th}$ and $(i+1)^{th}$ symbol of the common subsequence we try to find (and it plays some role only if that subsequence has length at least $i+1$); clearly, the longest common subsequence can be as long as the input words, so it has at most length $\min\{m,n\}$. We also consider the case when the gap constraints are given by the actual letters bounding the gap. All these extensions of \Problem{LCS} refer to models of constrained subsequences, where the constraints are local, i.\,e., they depend on the embedding of the symbols bounding the gap. Finally, we also consider \Problem{LCS} in the case of subsequences in bounded ranges, where the upper bound $B$ on the size of the ranges in which we look for subsequences is given as input; in this case, we have a global constraint on the embedding of the subsequence. 

After defining these variants of the {\Problem{LCS}}, which seem interesting and well motivated to us, we propose efficient algorithms for each of them. In most cases, these algorithms are non-trivial extensions of the standard dynamic programming algorithm solving {\Problem{LCS}}. A quick overview of our results for variants of \Problem{LCS}, when the gaps are local: if all gap constraints are identical or we have a constant number of different gap constraints (and the sequence of gap constraints fulfils some additional synchronization condition) we obtain algorithms running in $O(N)$ time; if we have arbitrarily many different constraints, we obtain an algorithm running in $O(N k)$, where $k $ is the length of the longest common constrained subsequence of the input words; if, moreover, the sequence of constraints is increasing, then the problem can be solved in $O(N \polylog N)$; if the constraints on the gaps are defined according to both letters bounding them, we obtain an algorithm running in $O(\min\{N \sigma \log N, N \sigma^2\})$; if the constraints on the gaps are defined according only to the letter coming after (respectively, before) them, we obtain an algorithm running in $O(\min\{N \log N, N \sigma\})$. In the case of subsequences in bounded range, we show an algorithm which runs in $O(NB^{o(1)})$ time for the respective extension of \Problem{LCS} (i.\,e., it runs in $O(NB^{d})$, for some $0<d<1$), as well as an $\frac{1}{3}-$approximation algorithm running in $O(N)$ time.\looseness=-1

\paragraph{Related Work.} With respect to algorithms, the results of \cite{IliopoulosEtAl2007} cover the case when all gap constraints are identical. In particular, \cite{IliopoulosEtAl2007} considers a variant of \Problem{LCS} where the lengths of the gaps induced by the embeddings of the common subsequence in the two input strings are all constrained by the same lower and upper bounds and, additionally, there is an upper bound on the absolute value of the difference between the lengths of the $i^{th}$ gap induced in $w$ and the $i^{th}$ gap induced in $v$, for all $i$. The authors of that paper propose a quadratic-time algorithm for this problem and then derive more efficient algorithms in some particular cases. To the best of our knowledge, the case of multiple gap-length constraints was not addressed so far in the literature. In \cite{AbboudEtAl2015}, the authors consider \Problem{LCS} for subsequences in a bounded range, called there \Problem{Local-2-LCS}, as an intermediate step in showing complexity lower bounds for \Problem{LCS}; they only mention the trivial $O(NB^2)$ algorithm solving it. The results of \cite{Charalampopoulos21} lead to an $O(N^{1+o(1)})$ solution for this problem; our solution builds on that approach. \looseness=-1

With respect to lower bounds, \Problem{LCS} is a particular case for all the problems we consider in our paper, as we can simply take all the length constraints to be trivial: $(0,n)$ in the case of gap constraints or $B=n$ in the case of subsequences in bounded ranges. Therefore, for each of our problems, the existence of an algorithm whose worst case complexity is $O(N^{1-c})$, with $c>0$, would refute SETH. Thus, if $\sigma\in O(1)$ (i.\,e., when the input is over alphabets of constant size), most of our algorithms solving \Problem{LCS} with gap-length constraints are optimal (unless SETH is false) up to polylog-factors; the exceptions are the two cases when we do not impose any monotonicity or synchronization condition on the tuple of gap-constraints. If $\sigma$ is not constant, the previous claim also does not hold anymore for the case when the constraints on the gaps are defined according to both letters bounding them. In the case of \Problem{LCS} for subsequences in a bounded range, \cite{AbboudEtAl2015} shows quadratic lower bounds even for $B$ being polylogarithmic in $n$. Thus, unless $B\in O(\polylog N)$, there is a super-logarithmic mismatch between the upper bound provided by our algorithm and the existing lower bound. 

It is natural to ask what happens when we have non-trivial constraints, such as, e.g., constraints of the form $(a,b)$ with $a,b\in O(1)$. In \cite{DayKMS22}, it is shown that deciding whether there exists an embedding of a string $v$ as a subsequence of another string $w$, such that this embedding satisfies a sequence of $|v|$ constraints of the form $(a,b)$ with $a,b\leq 6$, cannot be done in $O(N^{1-c})$ time, with $c>0$, unless SETH is false; moreover, the respective reduction can be modified so that the embedding fulfils our synchronization property for the case of $O(1)$ distinct gap constraints. The respective decision problem can also be solved by checking whether the longest common constrained subsequence of $v$ and $w$, where the gap-length constraints for the common subsequence are exactly those defined for the embedding of $v$ in $w$, has length $|v|$. So, the same lower bound from \cite{DayKMS22} (which coincides with the lower bound for the classical \Problem{LCS} problem) holds for the constrained \Problem{LCS} problem, even when we have a constant number of constant gap-length constraints, fulfilling, on top, the aforementioned synchronization property also. 
Due to page limitations some proofs are omitted in this version; see the full version \cite{AKKMS23}. \looseness=-1

%% file: preliminaries.tex
%%%%%%%%%%%%%%%%%%%%%%%%%%%%%%%%%%%%%%%%%%%%%%%%%%%%%%%%%%%%%%%%%%%%%%%%%%%%%%%%
%% Preliminaries
%%%%%%%%%%%%%%%%%%%%%%%%%%%%%%%%%%%%%%%%%%%%%%%%%%%%%%%%%%%%%%%%%%%%%%%%%%%%%%%%
\vspace*{-0.1cm}
\section{Preliminaries}\label{sec:prel}
\vspace*{-0.1cm}

Let $\mathbb{N} = \{1, 2, \ldots\}$ be the set of natural numbers, $[n] = \{1, \ldots, n\}$, and $[m:n]=[n]\setminus [m-1]$, for $m,n \in \mathbb{N}$. For $(a,b),(c,d)\in \mathbb{N}^2$, we write $(a,b) \subseteq (c,d)$ if and only if $c\leq a$ and $b\leq d$. All logarithms used in this paper are in base $2$. 

For a finite alphabet $\Sigma$, $\Sigma^+$ denotes the set of non-empty words (or strings) over $\Sigma$ and $\Sigma^* = \Sigma^+ \cup \{\emptyword\}$ (where $\emptyword$ is the empty word). For a word $w \in \Sigma^*$, $|w|$ denotes its length (in particular, $|\emptyword| = 0$).
%; for every $b \in \Sigma$, $|w|_{b}$ denotes the number of occurrences of $b$ in $w$. 
We set $w^0 = \emptyword$ and $w^k = w w^{k-1}$ for every $k \geq 1$. For a word $w$ of length $n$ and some $i\in [n]$, we denote by $w[i]$ the letter on the $i^{th}$ position of $w$, so $w = w[1] w[2] \cdots w[n]$. For every $i, j \in [|w|]$, we define $w[i:j] = w[i] w[{i+1}] \ldots w[j]$ if $i\leq j$, and $w[i:j]=\emptyword$, if $i>j$. For $w \in \Sigma^*$, we define $\alphabet{w} = \{b \in \Sigma \mid b$ occurs at least once in $w\}$. The strings $w[i:j]$ are called \emph{factors} of the string $w$; if $i=1$ (respectively, $j=|w|$), then $w[i:j]$ is called a \emph{prefix} (respectively, {\em suffix}) of $w$. For simplicity, for a word $w$ and two natural numbers $m\leq n$, we write $w\in [m,n]$ if $m\leq |w|\leq n$.

For an $m\times n$ matrix $M=(M[i,j])_{i\in [m],j\in [n]}$ and two sets $I\subset [m], J\subset [n]$, let $M[I,J]$ be the submatrix $(M[i,j])_{i\in I,j\in J}$ consisting in the elements which are at the intersection of row $M[i,\cdot]$ and column $M[\cdot,j]$ for $i\in I$ and $j\in J$. \looseness=-1

Further, we define the notions of subsequence and subsequence with gap-length constraints, following \cite{DayKMS22}. Our definitions are based on the notion of {\em embedding}. For a string $w$, of length $n$, and a natural number $k\in [n]$, an embedding is a function $e \colon [k] \to [n]$ such that $i < j$ implies $e(i) < e(j)$ for all $i, j \in [k]$. We say $e$ is a \emph{matching} embedding if $e(k)=n$. For strings $u,w\in\Sigma^\ast$ with $\len u\leq\len w$, an embedding $e \colon [\len u] \to [\len w]$ is an \emph{embedding of $u$ into $w$} if $u=w[e(1)]w[e(2)]\ldots w[e(k)]$, then $u$ is called a \emph{subsequence of $w$}. 

For an embedding $e \colon [k] \to [|w|]$ and every $j \in [k-1]$, the \emph{$j^{\text{th}}$ gap of $w$ induced by $e$} is the string $\gap{w}{e}{j} = w[e(j)+1:e(j+1)-1]$. A \emph{$t$-tuple of gap-length constraints} is a $t$-tuple $\gaptuple = (C_1, C_2, \ldots, C_{t})$ with $C_i =(\ell_i,u_i)$ and $0\leq \ell_i\leq u_i\leq n$ for every $i \in [t]$. We set $\gaptuple[i] = C_i$ for every $i \in [t]$, and $\gaptuple[1:i]= (C_1, C_2, \ldots, C_{i})$. We say that an embedding $e$ \emph{satisfies a $(k-1)$-tuple of gap-length constraints $\gaptuple$ with respect to a string $w$} if it has the form $e \colon [k] \to [|w|]$, and, for every $i \in [k - 1]$, $\ell_i \leq |\gap{w}{e}{i}| \leq u_i$ (that is, $\gap{w}{e}{i} \in C_i$). 

If there is an embedding $e$ of $u$ into $w$ satisfying the gap constraints $\gaptuple$, we denote this by $u\subseq_{\gaptuple} w$. For a $(k-1)$-tuple $\gaptuple$ of gap constraints, let $\subseqSet{\gaptuple}{w}$ be the set of all subsequences of $w$ induced by embeddings satisfying $\gaptuple$, i.\,e., $\subseqSet{\gaptuple}{w} = \{u \mid u\subseq_{\gaptuple} w \}$. The elements of $\subseqSet{\gaptuple}{w}$ are also called the \emph{$\gaptuple$-subsequences of $w$}. For more details see \cite{DayKMS22}.

We are interested in defining and investigating the longest common subsequence problem (\Problem{LCS} for short) in the context of subsequences with gap constraints. However, in the framework introduced in \cite{DayKMS22}, the gap constraints depend on (the length of) the subsequence, and this is not known for the \Problem{LCS} problem. As such, we propose a series of problems where we introduce variants of \Problem{LCS} accommodating gap-length constraints. In all our problems, we have two input strings $v$ and $w$, with $|v|=m$ and $|w|=n$ and $m\leq n$, and these strings are over an alphabet $\Sigma=\{1,2,\ldots,\sigma\}$, with $\sigma \leq m$. For the rest of this paper, let $N=mn$. We also consider w.l.o.g. that, when the input contains gap-length constraints, every individual constraint $C=(\ell,u)$ fulfils $0\leq \ell\leq u\leq n$. 

First, some additional definitions. Let $v,w\in\Sigma^\ast$ be two words; a word $s$ is a common subsequence of $v$ and $w$ if $s$ is a subsequence of both $v$ and $w$. Let $\gaptuple$ be a $(k-1)$-tuple of gap-length constraints. A word $s$ of length $k$ is a \emph{common $\gaptuple$-subsequence of $v$ and $w$} if both $s\subseq_{\gaptuple} w$ and $s\subseq_{\gaptuple} v$ hold. 
Let $\comSubseqSet{\gaptuple}{v}{w}$ denote $\subseqSet{\gaptuple}{v} \cap \subseqSet{\gaptuple}{w}$. % the set of all common $\gaptuple$-subsequences of $v$ and $w$. 

A $(k-1)$-tuple $\gaptuple$ of gap-length constraints is called {\em increasing} if $\gaptuple[i]\subseteq \gaptuple[i+1]$, for all $i\in [k-2]$. Let $\gaptuple$ be an increasing $(k-1)$-tuple of gap-length constraints and let $i\in [k-2]$. Assume $s$ is a $\gaptuple[1:i]$-subsequence embedded in $w[1:i']$, such that the last position of $s$ is mapped to $i'$, and $t$ is a $\gaptuple[1:j]$-subsequence embedded in $w[1:i']$ as well, such that the last position of $t$ is mapped to $i'$, and $j>i$. If  there exists $a\in \Sigma$ such that the embedding of $s$ in $w[1:i']$ can be extended to an embedding of $sa$ in $w[1:i'']$, which satisfies $\gaptuple[1:i+1]$, for some $i''>i'$, then the embedding of $t$ in $w[1:i']$ can be extended as well to an embedding of $ta$ in $w[1:i'']$ which satisfies $\gaptuple[1:i+1]$. \looseness=-1

A $(k-1)$-tuple $\gaptuple$ of gap-length constraints is called {\em synchronized} when it satisfies the property that for all $i,j\in [k-1]$, if $\gaptuple[i]=\gaptuple[j]$ and $i\leq j$ then $\gaptuple[i+e]\subseteq \gaptuple[j+e]$ for all $e\geq 0$ such that $i+e\leq j+e\leq m-1$; for example, the tuple $((0,5)(0,1) (0,2) (0,3) (0,1) (0,5) (0,3) (0,4))$ is synchronized. Let $\gaptuple$ be a synchronized $(k-1)$-tuple of gap-length constraints and let $i\in [k-2]$. Assume $s$ is a $\gaptuple[1:i]$-subsequence embedded in $w[1:i']$, such that the last position of $s$ is mapped to $i'$, and $t$ is a $\gaptuple[1:j]$-subsequence embedded in $w[1:i']$ as well, such that the last position of $t$ is mapped to $i'$, and $j>i$ and $\gaptuple[i+1]=\gaptuple[j+1]$. Now, if  there exists a letter $a$ such that the embedding of $s$ in $w[1:i']$ can be extended to an embedding of $sa$ in $w[1:i'']$ which satisfies $\gaptuple[1:i+1]$, for some $i''>i'$, then the embedding of $t$ in $w[1:i']$ can be extended as well to an embedding of $ta$ in $w[1:i'']$ which satisfies $\gaptuple[1:i+1]$.

The Longest Common Subsequence Problem ($\Problem{LCS}$) is defined as follows.
\begin{problem}[\Problem{LCS}]
Given $v,w$, compute the largest $k\in\mathbb [m]$ such that there exists a common subsequence $s$ of both $v$ and $w$ with $|s|=k$. 
\end{problem}

We now extend ${\Problem{LCS}}$ to the case of subsequences with gap constraints (for a more detailed discussion on variants of this problem, see the full version \cite{AKKMS23}). Firstly we consider the case when the constraints are {\em local}, as in \cite{DayKMS22}: they concern only the gaps occurring between two consecutive symbols of the subsequence.

\begin{problem}[$\Problem{LCS-MC}$]
Given $v,w\in\Sigma^\ast$ and an $(m-1)$-tuple of gap-length constraints $\gaptuple$, compute the largest $k\in\mathbb N$ such that there exists a common $\gaptuple[1:k-1]$-subsequence $s$ of $v$ and $w$, with $|s|=k$. That is, find the largest $k$ for which $\comSubseqSet{\gaptuple[1:k-1]}{v}{w}$ is non-empty.
\end{problem}
Clearly, $\Problem{LCS}$ is a particular case of $\Problem{LCS-MC}$, where $\gaptuple=((0,n),\ldots ,(0,n))$. 

In $\Problem{LCS-MC}$ the input tuple gap-length constraints contains arbitrarily many constraints (therefore the acronym MC in the name of the problem), as many as the maximum amount of gaps that a common subsequence of $v$ and $w$ may have (that is, $m-1$). We also consider \Problem{LCS-MC} for increasing tuples of gap-length constraints $\gaptuple$; this variant is called $\Problem{LCS-MC-INC}$.

We consider two special cases of Problem $\Problem{LCS-MC}$, where all these constraints are either identical (i.\,e., LCS with one constraint) or drawn from a set of constant size (i.e., LCS with $O(1)$ constraints), which seem interesting to us. 
\begin{problem}[\Problem{LCS-1C}]
Given $v,w\in\Sigma^\ast$ and an $(m-1)$-tuple of identical gap-length constraints $\gaptuple=((\ell,u),\ldots,(\ell,u))$, compute the largest $k\in\mathbb N$ such that there exists a common $\gaptuple[1:k-1]$-subsequence $s$ of $v$ and $w$, with $|s|=k$. \looseness=-1
\end{problem}

\begin{problem}[$\Problem{LCS-O(1)C}$]
Given $v,w\in\Sigma^\ast$ and an $(m-1)$-tuple of gap-length constraints $\gaptuple=((\ell_1,u_1),\ldots,(\ell_{m-1},u_{m-1}))$, where $|\{(\ell_i,u_i)\mid i\in [m-1]\}|$ (the number of distinct constraints of $\gaptuple$) is in $O(1)$, compute the largest $k\in\mathbb N$ such that there exists a common $\gaptuple[1:k-1]$-subsequence $s$ of  $v$ and $w$, with $|s|=k$. \looseness=-1
\end{problem}
The general results obtained for \Problem{LCS-MC} are improved for  \Problem{LCS-O(1)C}, by considering the latter problem in the restricted setting of synchronized gap-length constraints only. The resulting problem is called \Problem{LCS-O(1)C-SYNC}.

In the problems introduced so far, the gap between two consecutive symbols in the searched subsequence depends on the positions of these symbols inside the respective subsequence (i.\,e., the gap between the $i^{th}$ and ${i+1}^{th}$ symbols is always the same). That is, the actual symbols of the subsequence play no role in defining the constraints; it is only the length of the subsequence which is important. For the next problems, the constraints on a gap between consecutive symbols are determined by one or both symbols bounding the respective gap, and do not depend on the position of the gap inside the subsequence. 

For this, we first need to modify our setting. 
%As before, we have two input strings $v$ and $w$, with $|v|=m$ and $|w|=n$ and $m\leq n$, and these strings are over an alphabet $\Sigma=\{1,2,\ldots,\sigma\}$, with $\sigma \leq m$. 
Let $left:\Sigma\rightarrow [n]\times [n]$ and $right:\Sigma\rightarrow [n]\times [n]$ be two functions, defining the gap constraints. For an embedding $e \colon [k] \to [n]$, we say that $e$ \emph{satisfies the gap constraints defined by $(left,right)$ with respect to a string $x$} if for every $i \in [k - 1]$ we have that $\gap{x}{e}{i} \in left(x[e(i)]) \cap right(x[e(i+1)])$; in other words, $\gap{x}{e}{i}$ has to simultaneously fulfil the constraints $left(x[e(i)])$ and $right(x[e(i+1)])$, defined by the symbols bounding that gap. If there is an embedding $e$ of a string $y$ into $x$ satisfying the gap constraints $(left,right)$, we denote this by $y\subseq_{left,right} x$ and call $y$ a \emph{$(left,right)$-subsequence of $x$}. In the following algorithmic problems, functions $g:\Sigma \rightarrow [n]\times [n]$ are given as sequences of $\sigma$ tuples $(a,g(a))_{a\in\Sigma}$. 

\begin{problem}[\Problem{LCS-$\Sigma$}]
Given two words $v,w\in\Sigma^\ast$ and two functions $left:\Sigma\rightarrow [n]\times [n]$ and $right:\Sigma\rightarrow [n]\times [n]$, compute the largest number $k\in\mathbb N$ such that there exists a common $(left,right)$-subsequence $s$ of $v$ and $w$, with $|s|=k$. 
\end{problem}
When $left(a)=(0,n)$ for all $a\in \Sigma$ (respectively,  $right(a)=(0,n)$ for all $a\in \Sigma$), the gap constraints are defined only by the function $right$ (respectively, $left$), and the problem \Problem{LCS-$\Sigma$} is denoted \Problem{LCS-$\Sigma$R} (respectively, \Problem{LCS-$\Sigma$L}).

In the problems introduced so far, the constraints were local (in the sense that they were defined by consecutive problems in the subsequence). In the last problem we introduce, we build on the works \cite{AbboudEtAl2015,KoscheKMP22}, and consider subsequences which occur inside factors of bounded length of the input words. In particular, for a given integer $B$, a word $s$ is a $B$-subsequence of $w$ if there exists a factor $w[i+1:i+B]$ of $w$ containing $s$ as subsequence, and we look for the largest common $B$-subsequence of two input words. 
This problem was called Local-2-Longest Common Subsequence in \cite{AbboudEtAl2015}, but as the constraint acts now globally on the subsequence, we prefer to call it \Problem{LCS-BR} (LCS in bounded range), 
\begin{problem}[\Problem{LCS-BR}]
Given $v,w\in\Sigma^\ast$ and $B\in [n]$, compute the largest $k\in\mathbb N$ such that there exists a common $B$-subsequence $s$ of $v$ and $w$, with $|s|=k$. 
\end{problem}

We note that in all our definitions we are given a single tuple of gap-length constraints, meaning that the embeddings of the common subsequence of $v$ and $w$ should both fulfil the same constraints. Alternatively, we could have as input one tuple $\gaptuple_v$ of gap-length constraints for $v$ and one tuple $\gaptuple_w$ for $w$, constraining the embeddings of the common subsequence in $v$ and $w$, respectively. In this settings, the embeddings would depend both on the subsequence and on the target word, not only on the subsequence, as in the model used currently in the paper. 

Firstly, let us note that the model in which we have a single tuple of gap constraints seems more natural to us, as the gaps allowed in an embedding of a subsequence on a word seem to correspond rather to (or be determined by) properties of the subsequence, not to properties of the text in which it is embedded. For instance, in \cite{DayKMS22} as well as in the work on which that paper builds \cite{KleestMeissnerEtAl2021}, the gaps are defined for the string which one wants to embed as a subsequence in a larger string.

Secondly, most of our results hold as such for the case when we are given as input two sets of gap constraints instead of a single one. The only results that do not hold in an identical form are those which rely on 2D RMQ data structures, namely the solutions for \Problem{LCS-$\Sigma$L/R} running in $O(N \log N)$ and the solution for \Problem{LCS-$\Sigma$} running in $O(N \sigma \log N)$; in all these cases we need to extend the 2D RMQ structure to allow queries on rectangular submatrices (instead of quadratic only), and this leads to an increase in the complexities by a $\log n$-factor.\\ 

%\section{Computational Model}
%\label{sec:compmod}
We briefly discuss the {\em computational model} we use to describe our algorithms, solving efficiently the problems described in \Cref{sec:prel}. This model is the standard unit-cost RAM with logarithmic word size: for an input of size $L$, each memory word can hold $\log L$ bits. Arithmetic and bitwise operations with numbers in $[1:L]$ are, thus, assumed to take $O(1)$ time. Moreover, the numbers we are given as inputs (describing, e.\,g., the gap constraints) are given in binary encoding. In all the problems, we assume that we are given two words $w$ and $v$, with $|w|=n$ and $|v|=m$ (so the size of the input is $L = n+m$), over an alphabet $\Sigma=\{1,2,\ldots,\sigma\}$, with $2\leq |\Sigma|=\sigma\leq m$. That is, we assume that the processed words are sequences of integers (called letters or symbols), each fitting in $O(1)$ memory words. This is a common assumption in string algorithms: the input alphabet is said to be {\em an integer alphabet}. Moreover, as the problems deal with common subsequences, we can assume without loss of generality that the alphabets of the two words are identical. For more details  on this computational model see, e.\,g.,~\cite{crochemore}.

%We briefly discuss the {\em computational model} we use to describe our algorithms, solving efficiently the above problems. This model is the standard unit-cost RAM with logarithmic word size: for an input of size $L$, each memory word can hold $\log L$ bits. Arithmetic and bitwise operations with numbers in $[1:L]$ are, thus, assumed to take $O(1)$ time. Moreover, the numbers we are given as inputs (describing, e.\,g., the gap constraints) are given in binary encoding. In all the problems, we assume that we are given two words $w$ and $v$, with $|w|=n$ and $|v|=m$ (so the size of the input is $L = n+m$), over an alphabet $\Sigma=\{1,2,\ldots,\sigma\}$, with $2\leq |\Sigma|=\sigma\leq m$. That is, we assume that the processed words are sequences of integers (called letters or symbols), each fitting in $O(1)$ memory words. This is a common assumption in string algorithms: the input alphabet is said to be {\em an integer alphabet}. Moreover, as the problems deal with common subsequences, we can assume without loss of generality that the alphabets of the two words are identical. For more details  on this computational model see, e.\,g.,~\cite{crochemore}. 

%% file: specc.tex
\section{\Problem{LCS} with local gap constraints}
\label{sec:genc}

\paragraph{An initial approach for \Problem{LCS-MC}.} For all variants of \Problem{LCS} where the constraints are local (i.\,e., they depend on the position of the gap in the subsequence, or on the letters bounding it), the sets of subsequences which are candidates for $s$ can be, in the worst case, of exponential size in $N$. Therefore, computing the respective sets for both input words, their intersection, and then finding the longest string in this intersection would result in an exponential time algorithm. \Problem{LCS} can be, however, solved by a dynamic programming approach (considered folklore) in $O(N)$ time. Similarly, \Problem{LCS-MC} (and its particular cases \Problem{LCS-1C} and \Problem{LCS-O(1)C}, as well as \Problem{LCS-$\Sigma$}) can also be solved by a dynamic programming approach, running in polynomial time. We describe this general idea for \Problem{LCS-MC} only (as it can be easily adapted to all other problems). This idea reflects, to a certain extent, a less efficient implementation of the folklore algorithm for $\Problem{LCS}$. \looseness=-1

For input strings $v,w$ and constraints $\gaptuple = (C_1,\ldots, C_{m-1})$ we define, for each $p\in [m]$, a matrix $M_p\in\mathbb \{0,1\}^{m\times n}$, where $M_p[i,j]=1$ if and only if there exists a 
%common $\gaptuple[1:p-1]$-subsequence 
string $s_p$ with $\len {s_p} = p$ 
%of $v[1:i]$ and $w[1:j]$, with $\len {s_p} = p$, 
and matching embeddings $e_v,e_w$, respectively into $v[1:i]$ and $w[1:j]$, satisfying $\gaptuple[1:p-1]$.
%and there are two  embeddings $e_v,e_w$ (satisfying $\gaptuple[1:p-1]$) of $s_p$ into $v$ and $w$ respectively such that $e_v(p)=i,e_w(p)=j$. 
We compute $M_1$ by setting $M_1[i,j]=1$ if and only if $v[i]=w[j]$. Then we compute $M_p$ recursively by dynamic programming: let $C_{p-1} = (\ell,u)$ and note that $M_p[i,j]= 1$ if and only if $v[i]=w[j]$ and there are positions $i'$ with $\ell\leq i-i'-1\leq u$ and $j'$ with $\ell\leq j-j'-1\leq u$ such that there is a string $s_{p-1}$ of length $p-1$ 
with matching embeddings into $v[1:i']$ and $w[1:j']$, respectively, satisfying $\gaptuple[1:p-2]$. %ending at position $i'$ in $v$ and $j'$ in $w$ respectively. 
That is $M_p[i,j]= 1$ if and only if there is a $1$ in the submatrix $M_{p-1}[I,J]$ with $I=[i-u-1:i-\ell-1]$ and $J=[j-u-1:j-\ell-1]$. In the end, the length $k$ of the longest common subsequence of $v$ and $w$ satisfying $\gaptuple$ equals the largest $p$ such that $M_p$ is not the $0$-matrix. A na\"ive implementation of this approach runs in $O(N^2k)$ time. 

A more efficient implementation is given in the following.

\begin{figure}
\begin{center}
\begin{tikzpicture}
\matrix [matrix of math nodes,left delimiter=(,right delimiter=),row sep=0.5mm,column sep=0.5mm] (m) {
    M[1,1] & * &  *  & * & * & * & M[1,j] & *  & M[1,m]\\
    * & * &   *  & * &  *  & * & *  & *  &  *\\
    * & * &  M[i_u,j_u]  & * & M[i_u,j_\ell] & * & * & *  & *\\
    * & * &   *  & * &  *  & * & *  & *  &  *\\
    * & * &  M[i_\ell,j_u]  & * & M[i_\ell,j_\ell] & * & * & *  & *\\
    * & * &   *  & * &  *  & * & *  & *  &  *\\
    M[1,j] & * &  *  & * & * & * & M[i,j] & *  & *\\
    * & * &   *  & * &  *  & * & *  & *  &  *\\
    M[n,1] & * &  *  & * & * & * & * & *  & M[n,m]\\};
%relevant maximum search submatrix surround
\draw[dashed] ($0.5*(m-2-2.south east)+0.5*(m-3-3.north west)$) -- ($0.5*(m-3-5.north east)+0.5*(m-2-6.south west)$);
\draw[dashed] ($0.5*(m-6-2.north east)+0.5*(m-5-3.south west)$) -- ($0.5*(m-6-6.north west)+0.5*(m-5-5.south east)$);
\draw[dashed] ($0.5*(m-2-2.south east)+0.5*(m-3-3.north west)$) -- ($0.5*(m-6-2.north east)+0.5*(m-5-3.south west)$);
\draw[dashed] ($0.5*(m-3-5.north east)+0.5*(m-2-6.south west)$) -- ($0.5*(m-6-6.north west)+0.5*(m-5-5.south east)$);
%relevant current i,j surround surround
\draw[] ($0.5*(m-6-6.south east)+0.5*(m-7-7.north west)$) -- ($0.5*(m-7-7.north east)+0.5*(m-6-8.south west)$);
\draw[] ($0.5*(m-8-6.north east)+0.5*(m-7-7.south west)$) -- ($0.5*(m-8-8.north west)+0.5*(m-7-7.south east)$);
\draw[] ($0.5*(m-6-6.south east)+0.5*(m-7-7.north west)$) -- ($0.5*(m-8-6.north east)+0.5*(m-7-7.south west)$);
\draw[] ($0.5*(m-7-7.north east)+0.5*(m-6-8.south west)$) -- ($0.5*(m-8-8.north west)+0.5*(m-7-7.south east)$);
%relevant current i,j surround surround
\draw[dotted] ($0.5*(m-1-6.north east)+0.5*(m-1-7.north west)$) -- ($0.5*(m-7-7.north west)+0.5*(m-6-6.south east)$);
\draw[dotted] ($0.5*(m-7-1.north west)+0.5*(m-6-1.south west)$) -- ($0.5*(m-7-7.north west)+0.5*(m-6-6.south east)$);
\end{tikzpicture}
\end{center}
\caption{The computation of $M[i,j]$ in Lemma \ref{lem:specc} with $x_u=x-u-1$ and ${x_\ell=x-\ell-1}$ for $x\in\{i,j\}$.}
\label{fig:dynpro}
\end{figure}

\begin{lemma}\label{lem:specc}
\Problem{LCS-MC} can be solved in $O(Nk)$ time, where $k$ is the largest number for which there exists a common $\gaptuple[1:k-1]$-subsequence $s$ of $v$ and $w$.
\end{lemma}
\begin{proof}As mentioned above, we can compute $M_1$ in $O(N)$ time. So, let $2\leq p\leq m,$ $C_{p-1}=(\ell,u),$ and $d=\len {C_{p-1}} =u-\ell+1$. We want to compute the elements of $M_p$ and assume that $M_{p-1}$ was already computed. For convenience we treat $M_{p-1}[i,j]=0$ when either $i<1$ or $j<1$. 

We use a pair of $m\times n$ matrices $A$ and $B$, where $A[i,j]$ stores the sum of (or equivalently the amount of $1$s in) $d$ consecutive entries $M_{p-1}[i,j-d+1], \ldots,M_{p-1}[i,j]$ in the rows of $M_{p-1}$. Then 
$A[i,1]= M_{p-1}[i,1]$ and $A[i,j] = A[i,j-1] - M_{p-1}[i,j-d] + M_{p-1}[i,j]$ for all $i\in [m]$ and $j\in [2:n]$. The entry $B[i,j]$ stores the sum of all entries $M_{p-1}[i',j']$ with $0\leq i-i' < d$ and $0\leq j-j' < d$. Again, for convenience, we treat all entries $A[i,j]$ as $0$ if either $i<1$ or $j<1$. Then we compute $B[i,j]$ as follows. 
We set $B[1,1]=M_{p-1}[1,1]$, $B[1,j]=B[1,j-1] - M_{p-1}[1,j-d] + M_{p-1}[1,j]$, and $B[i,j]=B[i-1,j] - A[i-d,j] + A[i,j]$ for all $i\in [2:m]$ and $j\in[2:n]$.
%\[C[i,j]=
%\begin{cases*}
%	M_{p-1}[i,j] & if  $i=1=j$,\\
%	C[i,j-1] - M_{p-1}[i,j-d] + M_{p-1}[i,j] & if $i=1,1<j\leq n$,\\
%	C[i-1,j] - A[i-d,j] + A[i,j] & else. 
%\end{cases*}\]

Since the computation of each entry in $A$ or $B$ takes $O(1)$ time, we can compute the matrices $A$ and $B$ in time $O(N)$. Now $M_p[i,j]=1$ if and only if there is a $1$ in the submatrix $M_{p-1}[I,J]$, which is true if $B[i-\ell-1,j-\ell-1]>0$. Hence we can compute $M_p[i,j]$ in constant time, $M_p$ in $O(N)$ time, and the sequence $M_1,\ldots,M_{k+1}$ in time $O(Nk)$. 
\qed\end{proof}

\paragraph{An $O(N\log^2 N)$ time algorithm for \Problem{LCS-MC-INC}.}
We now consider \Problem{LCS-MC-INC}, a variant of \Problem{LCS-MC} where the tuple $\gaptuple$ is increasing. We begin with a lemma describing a data structure, which is then used to solve \Problem{LCS-MC-INC}.\looseness=-1
\begin{lemma}\label{lem:2DSegmentTrees}
Given an $m\times n$ matrix $M$ with all elements initially equal to $0$, we can maintain a data structure (two dimensional segment tree) ${\mathcal T}$ for $M$, so that we can execute the following operations efficiently:
\begin{itemize}
\item update$_{\mathcal T}(i',i'',j',j'',x)$: set $M[i,j]=\max\{M[i,j],x\}$, for all $i\in [i':i'']$ and $j\in [j':j'']$; here, $x$ is a natural number. Time: $O(\log n \log m)$. 
\item query$_{\mathcal T}(i,j)$: return $M[i,j]$. Time: $O( \log m)$. 
\end{itemize}
\end{lemma}
\begin{proof}[Sketch] The idea is to maintain a two dimensional (2D, for short) Segment Tree ${\mathcal T}$ for $M$ (see, for instance, \cite{deBerg2008,BergCKO08} for details about segment trees). The 2D Segment Tree ${\mathcal T}$ is defined for the matrix $M$ as follows (see, e.g., \cite{LauR21,IbtehazKR21}, and also note that this is a relatively standard data structure in competitive programming). 
\begin{itemize}
\item We define a segment tree $T$ for the range $[1:n]$. 
\item In each node $\alpha$ of the segment tree $T$ we have a segment tree $T_\alpha$ for the range $[1:m]$.
\item In each node $\beta$ of the tree $T_\alpha$ we store an integer value $val(\beta)$, which is initially $0$. 
\end{itemize}
Note that the nodes of $T$ correspond to sub-ranges $[x:y]$ of the range $[1:n]$. So, assume that we have a node $\alpha$ which corresponds to the range $[a,b]$. Then the nodes of $T_{\alpha}$ correspond canonically to the range $[a:b]$ (as they depend on $\alpha$) but also on a range $[c:d]$ of $[1:m]$ (as they are segment trees for $[1:m]$). So, the nodes of $T_{\alpha}$ correspond to submatrices $M[c:d][a:b]$ of $M$. 

Also there is a bijection between the leaves of the trees $T_{\alpha}$, where $\alpha$ is a leaf of $T$, and the elements $M[i,j]$, with $i\in [m]$ and $j\in [n]$. When constructing the structure ${\mathcal T}$ as above, we can compute and store for each $M[i,j]$ a pointer to the leaf corresponding to it. 

Now, we explain how the operations are performed (without going into details w.r.t. the standard usage of segment trees and the results regarding them). 

Consider first update$_{\mathcal T}(i',i'',j',j'',x)$. We first use the segment tree $T_1$ to identify the nodes $\alpha_1, \ldots, \alpha_e$, with $e\leq \log n$, which correspond to a partition of the interval $[j':j'']$. The process of identifying these nodes is implemented in the standard way, and requires $O(\log n)$ time. Then, for each of the trees $T_{\alpha}$, with $\alpha\in \{\alpha_1, \ldots, \alpha_e\}$, we  identify the nodes $\beta^{\alpha}_1, \ldots, \beta^{\alpha}_{e_\alpha}$, with $e_\alpha\leq \log m$, which correspond to a partition of the interval $[i':i'']$. Again, these can be identified in $O(\log m)$ time. Now, the nodes $\beta^{\alpha}_1, \ldots, \beta^{\alpha}_{e_\alpha}$, for $\alpha\in \{\alpha_1, \ldots, \alpha_e\}$, correspond to a set of $O(\log n \log m)$ submatrices which partition the submatrix $M[i':i''][j':j'']$ (whose elements we need to update). Further, for $\alpha\in \{\alpha_1, \ldots, \alpha_e\}$ and for $\beta\in \{\beta^{\alpha}_1, \ldots, \beta^{\alpha}_{e_\alpha}\}$, we set $val(\beta)=\max\{val(\beta),x\}$. 

The time complexity of this update operation is $O(\log n\log m)$. 

To retrieve the current value of $M[i,j]$, and answer query$_{\mathcal T}(i,j)$, we proceed as follows. Note first that the actual entry $M[i,j]$ might not have been changed. Therefore, we need to account for the updates we did on the trees. However, this is not complicated. We retrieve the leaf which corresponds to $M[i,j]$ (say that this is a leaf in a node $T_\alpha$). Then, we move up in the tree until we reach the root of this tree, and compute a value $ret$. Initially, $ret=M[i,j]$. Then, when we reach node $\beta$ of $T_\alpha$, we update $ret\gets \max\{ret,val(\beta)\}$. After processing the root of $T_\alpha$ we stop, and return $ret$ as the correct value of $M[i,j]$. 

The time complexity of this query operation is $O(\log m)$. 

The correctness of this approach follows immediately from the properties of segment trees. 
\qed\end{proof}

Based on this data structure, we can show the following Lemma. 
\begin{lemma}\label{lem:inc}
\Problem{LCS-MC-INC} can be solved in $O(N \log^2 N)$.
\end{lemma}
\begin{proof}
{\em Main idea.} Our algorithm computes, one by one, the elements of an $m\times n$ matrix~$M$ (whose elements are initially set to $0$). The approach is to define $M[i,j]$, for each pair of positions $(i,j)\in [m]\times [n]$ such that $v[i]=w[j]$, to equal the length $p$ of the longest string $s_p$ which has matching embeddings into $v[1:i]$ and $w[1:j]$, respectively, satisfying $\gaptuple[1:p-1]$. Because $\gaptuple$ is increasing (for all $i\in [m-2]$, $\gaptuple[i]\subseteq \gaptuple[i+1]$), $p$ can be determined as follows. It is enough to extend with the symbol $a=v[i]=w[j]$ (mapped to position $i$ of $v$ and position $j$ of $w$) the longest subsequence $s_{p'}$, with $|s_{p'}|=p'$, such that the embeddings of $s_{p'}$ in $v$ and $w$ end on positions $i'$ and $j'$, respectively, where the gap between $i'$ and $i$ and the gap between $j'$ and $j$ fulfil the gap constraint $\gaptuple[p'+1]$. Indeed, this is enough: as $\gaptuple$ is increasing, this longest subsequence can be extended in exactly the same way as any other shorter subsequence with the same properties. Then, to obtain $s_{p'}$, it is enough to set $p=p'+1$ and extend $s_{p'}$ with the letter $a$, mapped to $v[i]$ and $w[j]$ in the two embeddings, respectively. 

However, when considering the position $(i,j)$, we do not know the value of $p'$, and, as such, we do not know the range where we need to look for $i'$ and $j'$. Therefore, we need to find a way around this.

{\em Dynamic programming approach.} We now show how to compute the elements of $M$. In the case of the dynamic programming algorithms solving {\Problem{LCS}}, the element $M[i,j]$ of the matrix $M$ is computed by looking at some elements $M[i',j']$, with $i'\leq i, j'\leq j, (i,j)\neq (i',j').$ By the arguments presented above, such an approach does not seem to work directly for {\Problem{LCS-MC-INC}}. However, if we know the value $p$ of some entry $M[i,j]$, we can be sure that $M[i'',j'']\geq p+1$ for all $i''$ and $j''$ such that $i+\ell+1 \leq i'' \leq i + u +1$ and $j+\ell+1 \leq j'' \leq j + u +1$, where $\gaptuple[p+1]=(\ell,u)$; we store this information. Moreover, if we know already all the values $M[i,j]$, with $i\leq i', j\leq j', (i,j)\neq (i',j'),$ then we have already seen (and stored) all possible values for $M[i',j']$ (or, in other words, all possible subsequences that we can extend in order to get $M[i,j]$), so we simply set $M[i',j']$ to the largest such possible value. 

So, we compute the elements $M[i,j]$ one by one, by traversing the elements of $M$ for $i$ from $1$ to $m$, for $j$ from $1$ to $n$, and proceed as follows. When we reach an element $M[i,j]$ in our traversal of $M$, we simply set it permanently to its current value. Then, if we set $M[i,j]$ to some value $p$, and $\gaptuple[p]=(\ell,u)$, we update each element $M[i',j']$ of submatrix $M[I,J]$, where $I=[i+\ell+1 : i + u +1]$ and $J=[j+\ell+1 : j + u +1]$, to be $M[i',j'] = \max\{M[i',j'],p+1\}$.  

{\em The algorithm.} First we define the matrix $M$, and initialize all its entries with $0$. Then, we build the data structure ${\mathcal T}$ from Lemma \ref{lem:2DSegmentTrees} for $M$. In an initial step, we set all values $M[1,j]=1$, where $v[1]=w[j]$, and $M[i,1]=1$, where $v[i]=w[1]$; this is done by using update-operations on ${\mathcal T}$ (to set the entry $M[i,j]=x$, for some $x>0$, given that $M[i,j]$ was equal to $0$, it is enough to execute update$(i,i,j,j,x)$). Further, we execute the following procedure.
\begin{itemize}
\item[1:] for $i=2$ to $m$ do
\item[2:] \hspace*{.5cm} for $j=2$ to $n$ do
\item[3:] \hspace*{1cm} Set $M[i,j]=$query$_{\mathcal T}(i,j)=p$; $M[i,j]$ remains equal to $p$ permanently; 
\item[4:] \hspace*{1cm} For $(\ell,u)=\gaptuple[p]$ update$_{\mathcal T}(i+\ell+1, i + u +1, j+\ell+1 ,j + u +1, p+1)$. 
\end{itemize}
The solution to \Problem{LCS-MC-INC} is the maximum element of $M$.
 
{\em Conclusion.} The correctness of our algorithm follows from the arguments presented above. The time complexity of the algorithm is $O(N\log^2 N)$, as we need $O(N\log^2 N)$ time for the preprocessing (setting up ${\mathcal T}$ and doing the initial updates on it). Then the 4-step procedure described above takes also $O(N\log^2 N)$ time, as in each iteration of the inner loop we perform as the most time consuming operation an update on ${\mathcal T}$. Our claim follows.
\qed\end{proof}

Summing up, we have shown the following theorem regarding \Problem{LCS-MC}. 
\begin{theorem}
\Problem{LCS-MC} can be solved in $O(Nk)$ time, where $k$ is the largest number for which there exists a common $\gaptuple[1:k-1]$-subsequence $s$ of $v$ and $w$.  \Problem{LCS-MC-INC} can be solved in $O(N \log^2 N)$.
\end{theorem}

%% file: onec.tex
\paragraph{$O(N)$ solutions for \Problem{LCS-1C} and \Problem{LCS-O(1)C}.}
\label{sec:onec}
While this problem was already solved in \cite{IliopoulosEtAl2007}, we also briefly describe our solution for it. Our approach is based on the following data-structures lemma, which are also used to solve some of the other problems we discuss here.

\begin{lemma}
\label{lem:maxdeq}
Let $\Psi:[m]\times [n]\rightarrow \{0,1\}$ be a predefined function, such that $\Psi(i,j)$ can be retrieved in $O(1)$ time. Given an $m\times n$ matrix $M$, with all elements initially equal to $0$, and four positive integers $\ell_1\leq u_1,\ell_2\leq u_2$, we can maintain a data structure ${\mathcal D}$ for $M$, so that the following process runs in $O(N)$ time:
\begin{itemize}
\item[1:] for $i=1$ to $m$ do
\item[2:] \hspace*{.5cm} update ${\mathcal D}$ (set up for processing line $i$);
\item[3:] \hspace*{.5cm} for $j=1$ to $n$ do
\item[4:] \hspace*{1cm} update ${\mathcal D}$ (set up for computing $M[i,j]$);
\item[5:] \hspace*{1cm} use ${\mathcal D}$ to retrieve $\mathfrak m$, the maximum of the submatrix $M[I,J]$ \\
\hspace*{1.4cm}  where $I=[i-u_1:i-\ell_1]$ and  $J=[j-u_2:j-\ell_2]$; \\
 \hspace*{1cm} $\mathfrak m$ is set to be $0$ when $I$ or $J$ are empty. 
\item[6:] \hspace*{1cm} if $\Psi(i,j)=1$ then set $M[i,j]=\mathfrak m+1$.
\end{itemize}
\end{lemma}

\begin{proof}
{\em Preprocessing phase.} In the preprocessing, we define ${\mathcal D}$. This data structures contains a double ended queue (deque) $Q_f$ for each column $M[\cdot,f]$ of the matrix $M$, with $f\in [n]$, as well as an array $Max$ with $n$ elements. 

{\em Main idea.} In general we maintain the following invariant property for deque $Q_f$: the content of the deque $Q_f$ is the list of matrix entries (in order, from the first element in $Q_f$ to the last element of $Q_f$) $M[i_{1,f},f], \ldots, M[i_{e_f,f},f]$ of $M[\cdot,j]$ such that the following hold: 
\begin{itemize}
\item $i_{1,f}<i_{2,f}<\ldots <i_{e_f,f}$, 
\item $M[i_{1,f},f]> \ldots > M[i_{e_f,f},f]$, and 
\item $M[i_{g,f},f]>M[h,f]$ for all $i_{g-1,f}< h< i_{g,f}$, for $g\in [2:e_f]$.
\item After executing step $4$ of the process above for some values $i$ and $j$, the queues $Q_f$, with $f\in J$, only contain elements of the submatrix $M[I,J]$. 
\end{itemize}

{\em Implementation.} Initially, all deques in ${\mathcal D}$ are empty. 

The data structure ${\mathcal D}$ is updated as follows (in step $2$ of our process): for $f$ from $1$ to $n$, we remove from the deque $Q_f$ the element $M[i-1-d_1,f]$, if that was contained in $Q_f$. As we execute this step for all values $i$, it is clear that when we execute step $2$ for $i=a$, then the first element of $Q_f$ is  $M[e,f]$ for some $e\geq i-1-d_1$. Moreover, once this step is completed, we recompute the array $Max$ such that $M[f]$ is the maximum between the first element of $Q_f$ (i.\,e., the greatest element of $Q_f$) and $M[i-d_1+\ell_1,f]$. For the array $Max$ we construct in $O(n)$ time data structures allowing us to answer Range Maximum Queries in $O(1)$ time (see \cite{BenderF00}). At this point, we can retrieve in $O(1)$ time the maximum element in any subarray $Max[a:b]$, with $1\leq a\leq b\leq n$. 

Further, ${\mathcal D}$ is updated as follows (in step $4$ of our process).
\begin{itemize}
\item When we execute step $3$ of our algorithm for some values $i$ and $j$, we insert $M[i-d_1+\ell_1,j-d_2+\ell_2]$ in $Q_{j-d_2-\ell_2}$. 
\item The insertion of a value $x$ in the deque $Q_f$ for some $f$ is handled as follows: we traverse $Q_f$ from last element towards the first and remove all values smaller or equal to $x$. When we meet an element strictly greater than $x$ or we have emptied $Q_f$, we store the value $x$ as the last element of $Q_f$. This $x$ is now the smallest element of $Q_f$ (while the first element in $Q_f$ is the greatest element in $Q_f$). 
\end{itemize}
Note that after the execution of this update step, the first element of the deque $Q_f$ is exactly the element $Max[f]$, for all $f\in J$. 

So, using the Range Maximum Query structures constructed in step $2$, we can implement step $5$ as simply querying to find the greatest element of $Max$ over the range $J$. This is returned in $O(1)$ time.

{\em Conclusion.} The correctness of this implementation follows from the explanations given above. To analyse the complexity of the algorithm, we note first that step $2$ takes $O(n)$ time (and is executed $m$ times), steps $5$ and $6$ take $O(1)$ time (and are executed $mn$ times). To see how much time is spent in the execution of step $4$ over the entire execution of the algorithm it is enough to note that each element of $M$ is inserted once in one of the deques stored in ${\mathcal D}$, and the time spent in step $4$ is proportional to the number of elements removed from these deques. So, the time spent overall in the execution of this step is upper bounded by the number of elements inserted in the deques. Thus, step $4$ adds at most $O(mn)$ time to the overall complexity of the algorithm. In total, this means that the respective procedure runs $O(N)$ time, in the implementation described here.
\qed\end{proof}

We can now show immediately the following result.
\begin{theorem}
$\Problem{LCS-1C}$ can be solved in $\bigO(N)$ time.
\end{theorem}
\begin{proof}
Let $(\ell,u)$ be the single gap-length constraint appearing in $\gaptuple$. We define the $m\times n$ matrix $M$, where $M[i,j]=p$ if and only if $p$ is the greatest number for which there exists a subsequence $s_p$, with $|s_p|=p$, such that there are two matching embeddings $e_v$ and $e_w$ of $s_p$ into $v[1:i]$ and $w[1:j]$, respectively, both satisfying $\gaptuple[1:p-1]$. We have that $M[i,j]=p$ if and only if $v[i]=w[j]$ and $p-1$ is the greatest number for which there exist $i'$ and $j'$ with $i-u-1\leq i'\leq i-\ell-1$ and $j-u-1\leq j'\leq j-\ell-1$ and $M[i',j']=p-1$.
%; that is, there exists an embedding of subsequence of length $p-1$ which fulfils the gap constraints and which can be extended with a $p^{th}$ symbol that is embedded in $v[i]$ and $w[j]$, respectively. 
Hence, the entries of $M$ can be computed using Lemma \ref{lem:maxdeq}, for $u_1=u_2=u+1$, $\ell_1=\ell_2=\ell+1$ and $\Psi(i,j)=1$ if and only if $v[i]=w[j]$. \looseness=-1
\qed\end{proof}

This result can be extended to the case of \Problem{LCS-O(1)C-SYNC}. It is, however, open if a similar result holds for the unrestricted problem {\Problem{LCS-O(1)C}. 

\begin{theorem}\label{thm:sync}
\Problem{LCS-O(1)C-SYNC} can be solved in $\bigO(N)$ time, where the constant hidden by the $O$-notation depends linearly on the number $h$ of distinct gap-length constraints of $\gaptuple$.
\end{theorem}
The result of Theorem \ref{thm:sync} holds, in fact, for a larger family of constraints, namely constraints whose elements can be partitioned in $h\in O(1)$ classes, such that, for all $1\leq i<j\leq k-1$, if $i,j$ are in the same class of the partition then $\gaptuple[i]=\gaptuple[j]$ and $\gaptuple[i+e] \subseteq \gaptuple[j+e]$ for all $e\geq 0$ such that $i+e\leq j+e\leq m-1$. \looseness=-1

\begin{proof}
{\em Preprocessing phase.} Assume $(\ell'_1,u'_1), \ldots, (\ell'_h,u'_h)$ is an enumeration of the distinct constraints from the set $\{(\ell_i,u_i)\mid i\in [m-1]\}$, where $h$ is a constant and $(\ell'_g,u'_g)\leq (\ell'_{g+1},u'_{g+1})$ (w.r.t. canonical ordering of pairs of natural numbers). In $O(m+n)$ time, we can radix-sort the list constraints of $(\ell_i,u_i)$, with $ i\in [m-1]$, and obtain the list $(\ell'_1,u'_1), \ldots, (\ell'_h,u'_h)$, as well as an array $label$ with $m-1$ elements, where $label[i]=j$ if and only if $(\ell_i,u_i)=(\ell'_j,u'_j)$; we say, for simplicity, that $label[i]$ defines the gap $(\ell'_{label[i]}, u'_{label[i]})$. Now we move on to the description of the main algorithm. 

{\em Main idea}. Our approach is to compute for each pair of positions $(i,j)\in [m]\times [n]$, such that $v[i]=w[j]=a$, and each $r\in [h]$ the longest common subsequence $s_p$ of $v[1:i]$ and $w[1:j]$, with $|s_p|=p$ and $label(p)=r$, which fulfils $\gaptuple[1:p-1]$ and the last symbol of $s_p$ is mapped to $v[i]$ in the embedding of $s_p$ in $v[1:i]$ and to $w[j]$ by the embedding of $s_p$ in $w[1:j]$. This is obtained by extending a common subsequence $s_{p-1}$ of $v$ and $w$, with $|s_{p-1}|= p-1$, which fulfils $\gaptuple[1:p-2]$ and whose last symbol is mapped to position $i'$ of $v$ and to position $j'$ of $w$, such that the gap between $i'$ and $i$ and the gap between $j'$ and $j$ fulfil the gap constraint defined by $r'=label[p-1]$. 

The main observation here is that, due to the synchronization property of $\gaptuple$ (for all $i,j\in [m-1]$, if $\gaptuple[i]=\gaptuple[j]$ and $i\leq j$ then $\gaptuple[i+e]\subseteq \gaptuple[j+e]$ for all $e\geq 0$), $p$ can be determined as follows. For each $r'\in [h]$, it is enough to extend with a symbol $a$ (mapped to position $i$ of $v$ and position $j$ of $w$) only the longest subsequence $s_{m_{r'}}$ such that $label(m_{r'})=r'$ and the embeddings of $s_{m_{r'}}$ in $v$ and $w$ end on positions $i'_{r'}$ and $j'_{r'}$, respectively, where the gap between $i'$ and $i$ and the gap between $j'$ and $j$ fulfil the gap constraint defined by $r'$. Indeed, this is enough: due to the synchronization of $\gaptuple$, this longest subsequence can be extended in exactly the same way as any other shorter subsequence with the same properties. Then, to obtain $s_p$, it is enough to extend the longest of the subsequences $s_{m_{r'}}$, for $r'\in [h]$, such that $label[m_{r'}+1]=r$. This naturally leads to a dynamic programming algorithm, which we describe in the following.

{\em Dynamic programming.}  For each $r\in[h]$, we define an $m\times n$ matrix $M_{r}$, where $M_{r}[i,j]=p$ if and only if $v[i]=w[j]$ and $p$ is the greatest number for which $\label[p]=r$ and there exists a $\gaptuple[1:p-1]$-subsequence $s_p$ of $v[1:i]$ and $w[1:j]$ for which there are two embeddings $e_v$ and $e_w$ of $s_p$ into $v[1:i]$ and $w[1:j]$, respectively, with $e_v(p)=i,e_w(p)=j$. 

According to the definition of $M_r$ and the observations made above, we have that $M_r[i,j]=p$ if and only if $v[i]=w[j]$ and $p-1$ is the greatest number for which $label[p]=r$ and there exist $i'$ and $j'$ with $i-u'_{r'}-1\leq i'\leq i-\ell'_{r'}-1$ and $j-u'_{r'}-1\leq j'\leq j-\ell'_{r'}-1$  such that $M_{r'}[i',j']=p-1$, for some $r'\in [h]$. That is, $M_r[i,j]=p$ if there exists an embedding of subsequence of length $p-1$ which fulfils the gap constraints and which can be extended (while still fulfilling the constraints) to a subsequence of length $p$, whose last symbol is embedded in $v[i]$ and $w[j]$, respectively, and, moreover, the label of the $p^{th}$ gap (the one following the newly found $p^{th}$ symbol) is $r$. 

So, to compute $M_r[i,j]$ we first have to compute for all $r'\in [h]$ the maximum value $m_{r'}$ of $M_{r'}[I_{r'},J_{r'}]$, for $I_{r'}=[i-u'_{r'}-1: i-\ell'_{r'}-1]$ and $J_{r'}=[j-u'_{r'}-1: j-\ell'_{r'}-1]$. Then, we simply set $M_r[i,j]=1 + \max_{r'\in [h], label[m_{r'}+1]=r} m_{r'}$. 

Now, what remains to be explained is how to retrieve the maximum value $p_{r'}$ of $M_{r'}[I_{r'},J_{r'}]$, for $I_{r'}=[i-u'_{r'}-1: i-\ell'_{r'}-1]$ and $J_{r'}=[j-u'_{r'}-1: j-\ell'_{r'}-1]$. For this we use Lemma \ref{lem:maxdeq}, as shown below.

{\em The algorithm}. We initialize the $m\times n$ matrices $M_r$, for $r\in [h]$, such that all their entries are $0$, and define for each of them a data structure ${\mathcal D}_r$ as in Lemma \ref{lem:maxdeq}, with the four input numbers being $\ell'_r+1, u'_r+1,  \ell'_r+1, u'_r+1$. Then, we adapt the procedure of Lemma \ref{lem:maxdeq} to compute these matrices as follows:
\begin{itemize}
\item[1:] for $i=1$ to $m$ do
\item[2:] \hspace*{.5cm} update ${\mathcal D}_{r'}$, for $r'\in [h]$;
\item[3:] \hspace*{.5cm} for $j=1$ to $n$ do
\item[4:] \hspace*{1cm} update ${\mathcal D}_{r'}$, for $r'\in [h]$;
\item[5:] \hspace*{1cm} for $r'\in[h]$, use ${\mathcal D}_{r'}$ to retrieve $m_r$, the maximum of $M_{r'}[I_{r'},J_{r'}]$ \\
\hspace*{1.4cm}  where $I_{r'}=[i-u'_{r'}-1: i-\ell'_{r'}-1]$ and $J_{r'}=[j-u'_{r'}-1: j-\ell'_{r'}-1]$; \\
 \hspace*{1cm} $m_{r'}$ is set to be $0$ when $I_{r'}$ or $J_{r'}$ are empty. 
\item[6:] \hspace*{1cm}  Compute set $max_r=\max_{r'\in [h], label[m_{r'}+1]=r} m_{r'}$, for all $r\in [h]$; 
\item[7:] \hspace*{1cm} for $r\in[h]$, if $v[i]=w[j]$ then set $M_r[i,j]=1+max_r$.
\end{itemize}

The result of \Problem{LCS-O(1)-SYNC} is given by the maximum value stored in one of the matrices $M_r$ at the end of the computation. 

{\em Conclusion.} The correctness of the algorithm follows from the explanations given above. Its time complexity is $O(N)$, based on the result of Lemma \ref{lem:maxdeq} and on the fact that $h\in O(1)$. 
\qed\end{proof}

%\begin{algorithm2e}[!htb]
%%	\caption{\lcsoneProb}
%	\label{algo:lcsone}
%	
%	\SetAlgoLined
%	\KwIn{words $v$,$w \in \Sigma^*$, gap constraint $C \in \Sigma^*$}
%	\KwResult{length of the LCS of $v$ and $w$ with regard to $\gaptuple = (C_1, \ldots, C_{k-1})$ with $C_i = C$}
%	
%	\BlankLine
%	construct matrix $M$\;
%	
%	
%	\For{i=0 \KwTo $\len{v}$}{
%		\For{j=0 \KwTo $\len{w}$}{
%			%TODO
%		}
%	}
%	\Return $k$\;
%\end{algorithm2e}

%%%%%%%%%%%%%%%%%%%%%%%%%%%%%%%%%%%%%%%%%%%%%%%%%%%%%%%%%
% this is just some algorithm template as reference
%\begin{algorithm2e}[!htb]
%	\caption{Building the Simon-Tree $T_w$ for a word $w$}
%	\label{alg:simontree}
%	
%	\SetAlgoLined
%	\KwIn{Word $w$ with $|w| = n$}
%	\KwResult{Simon-Tree~$T_w$}
%	$w' \gets w\$ $\;
%	Let $T$ be the tree with the root associated to the block $[?:n+1]$ of $w'$\;
%	Let $p$ be a pointer to the root of $T$\;
%	Compute the array $X[i]$\;
%	\BlankLine
%	
%	\For{i=n \KwTo 1}{
	%		$a \gets$ \KwFindNode{$i$,$T$,$p$}\;
	%		$(T,p) \gets$ \KwSplitNode{$i$,$T$,$a$}\;
	%	}
%	\BlankLine
%	
%	Set starting position for all blocks from leftmost branch including the root to $1$\;
%	%	\tcp{remove $\$$-letter from tree}
%	Remove $\$$-letter from tree: Remove the node associated to $[n+1:n+1]$ from $T$ and set all right ends $r$ of blocks on the rightmost branch to $n$\;
%	\Return T \;
%\end{algorithm2e}

%% file: alphc.tex
First, we analyse the problem \Problem{LCS-$\Sigma$R}. The input of this problem consists in two words $v$ and $w$ and one function $right:\sigma \rightarrow [n]\times [n]$ with $right(a) = (\ell_a, u_a)$ for all $a \in \Sigma$ (that is, for the $right$-function we have $right(a)=(0,n)$ for all $a\in \Sigma$). The approach we use is based on Lemma \ref{lem:maxdeq}, as in the solution to \Problem{LCS-O(1)C-SYNC}. 

\begin{lemma}\label{lem:lcsSRsigma}
\Problem{LCS-$\Sigma$R} can be solved in $O(N\sigma)$ time.
\end{lemma}
\begin{proof}
{\em Main idea}. Our approach is to compute for each pair of positions $(i,j)\in [m]\times [n]$, such that $v[i]=w[j]=a$, the longest $(right)$-subsequence $s_p$ of both $v$ and $w$, where $|s_p|=p$ and the last symbol of $s_p$ is mapped to $v[i]$ in the embedding of $s_p$ in $v[1:i]$ and to $w[j]$ by the embedding of $s_p$ in $w[1:j]$. This can be obtained by simply extending with the symbol $a=v[i]=w[j]$ the longest $(right)$-subsequence $s_{r}$, of length $r$, whose last symbol is mapped to position $i'$ of $v$ and to position $j'$ of $w$, such that the gap between $i'$ and $i$ and the gap between $j'$ and $j$ fulfil the gap constraint defined by $right(a)$; clearly, $p$ is then defined as $r+1$. 

{\em Dynamic Programming.} We compute an $m\times n$ matrix $M$, where $M[i,j]$ is length of the longest $(right)$-subsequence $s_p$ of both $v$ and $w$, where $|s_p|=p$ and the last symbol of $s_p$ is mapped to $v[i]$ in the embedding of $s_p$ in $v[1:i]$ and to $w[j]$ by the embedding of $s_p$ in $w[1:j]$; $M[i,j]=0$ if and only if $v[i]\neq w[j]$. To compute $M[i,j]$ we need to retrieve the largest entry $M_a$ of the submatrix $M[I_a,J_a]$, where $I_a=[i-u_a-1:i-\ell_a-1]$ and $J_a=[j-u_a-1:j-\ell_a-1]$, and set $M[i,j]=M_a+1$. We can proceed as follows.

{\em The algorithm}. We initialize the $m\times n$ matrices $M$ such that all their entries are $0$, and define for each $a\in \Sigma$ a data structure ${\mathcal D}_a$ as in Lemma \ref{lem:maxdeq}, for the matrix $M$ with the four input numbers being $\ell_a+1, u_a+1,  \ell_a+1, u_a+1$, where $right(a)=(\ell_a,u_a)$. Then, we adapt the procedure of Lemma \ref{lem:maxdeq} to work as follows:
\begin{itemize}
\item[1:] for $i=1$ to $m$ do
\item[2:] \hspace*{.5cm} update ${\mathcal D}_{b}$, for $b\in \Sigma$;
\item[3:] \hspace*{.5cm} for $j=1$ to $n$ do
\item[4:] \hspace*{1cm} update ${\mathcal D}_{b}$, for $b\in '\Sigma$;
\item[5:] \hspace*{1cm} if $w[j]=v[i]$, let $a=w[j]$;
\item[5:] \hspace*{1cm} use ${\mathcal D}_{a}$ to retrieve $M_a$, the maximum of $M[I_{a},J_{a}]$ \\
\hspace*{1.4cm}  where $I_a=[i-u_a-1:i-\ell_a-1]$ and $J_a=[j-u_a-1:j-\ell_a-1]$; \\
 \hspace*{1cm} $M_{a}$ is set to be $0$ when $I_{a}$ or $J_{a}$ are empty. 
\item[6:] \hspace*{1cm}  set $M[i,j]=1+M_a$.
\end{itemize}

{\em Conclusion.} The correctness of the algorithm follows from the explanations given above. Its time complexity is $O(N\sigma)$, based on the result of Lemma \ref{lem:maxdeq} and on the fact that we need to maintain $\sigma$ data structures ${\mathcal D}_a$, for $a\in \Sigma$, and each of them can be maintained in overall time $O(N)$. 
\qed\end{proof}

We now present another algorithm for \Problem{LCS-$\Sigma$R}, running in $O(N \log m)$ time. 

This algorithm uses a special case of the two-dimensional Range Maximum Query data structure (for short, RMQ), which is able to answer maximum queries on square sized submatrices of a matrix $M$ of size $m \times n$ in $O(\log (m))$ time, and which allows a special type of updates needed in our solution for \Problem{LCS-$\Sigma$R}. This structure extends the Sparse Table approach from \cite{BenderF00}.

\paragraph{The two dimensional Range Maximum Query structure.}
%We first define the internal structure and explain how the queries on this data structure are implemented before using it to solve the problem \Problem{LCS-$\Sigma$R}. 

At its base, our RMQ data structure maintains an $m \times n \times (1+\lceil \log m \rceil)$ array $RMQ[\cdot,\cdot,\cdot]$, such that for $1 \leq i\leq m$ and $1\leq j \leq n$ and $0 \leq q \leq \lceil \log m \rceil$, $RMQ[i,j,q]$ stores the maximum of the submatrix $M[I,J]$ with $I= [i-2^q+1:i]$ and $J=[j-2^q+1:j]$. For simplicity, we assume that $M[i,j]=0$ if $i\leq 0$ or $j\leq 0$, and $RMQ[i,j,q]$ is defined as $0$ if $i\leq 0$ or $j\leq 0$. Further, in a linear time preprocessing phase we can compute an array $Q[1:n] $ where, for $h\in [m]$, we have $Q[h]=\max\{q \in \mathbb{N}\cup\{0\} \mid 2^q\leq h\}$ (in other words, $Q[h]=\lfloor \log h\rfloor$). Let us now see how to retrieve the maximum of an arbitrary square submatrix of $M[I,J]$ with $I=[i':i'']$ and $J=[j':j'']$ (with $|i''-i'| = |j''-j'|$) in constant time, once we have computed the three dimensional array $RMQ$.

\begin{lemma}
Given $RMQ[\cdot,\cdot,\cdot]$, we can retrieve $\max M[I,J]$ in $O(1)$.
\label{lem:rmqquery}
\end{lemma}
\begin{proof}
As said, assume that we want to compute the largest value of $M[I,J]$ with $I=[i':i'']$ and $J=[j':j'']$. At first we need to determine the largest $q$, such that a square of size $2^q$ completely fits into the square $M[I,J]$. That is to find a maximal $q$ with $2^q \leq |i'-i''| = |j'-j''| < 2^{q+1}$, so $q=Q[|i'-i''|]$.

We now claim that the maximum of the values $RMQ[i'+2^q,j'+2^q,q]$, $RMQ[i'+2^q,j'',q], RMQ[i'',j'+2^q,q], RMQ[i'',j'',q]$ is a maximum in $M[I,J]$. We distinguish two cases.

\textbf{Case 1:} $2^q = |i'-i''| = |j'-j''|$.
In this case we replace $2^q$ in all four cases by $|i'-i''|$ and get immediately four times the lookup $RMQ[i'',j'',q]$, that is by definition of $RMQ[\cdot,\cdot,\cdot]$ the desired answer.

\textbf{Case 2:} $2^q < |i'-i''| = |j'-j''| < 2^{q+1}$.
Let us have a look at the sets of values considered by the single $RMQ$ look up operations. We have $M[I_1,J_1], M[I_1,J_2], M[I_2,J_1], M[I_2,J_2]$ with $I_1=[i':i'+2^q-1]$, $J_1=[j':j'+2^q-1]$, $I_2=[i''-2^q+1:i'']$ and $J_2=[j''-2^q+1:j'']$ . Because we have $2^q < |i'-i''| < 2^{q+1}$ and $2^q < |j'-j''| < 2^{q+1}$, we have $I_1 \cap I_2 \neq \emptyset$ and $J_1 \cap J_2 \neq \emptyset$. Therefore, we completely cover all values of $M[I,J]$.
\qed\end{proof}

Now that we have understood how the array $RMQ$ is used, we see how to calculate $RMQ[i,j,q]$ for all $q$ initially. Clearly, $RMQ[i,j,0]=M[i,j]$, for all $i\in [m]$ and $j\in [n]$. Then we use a dynamic programming approach. We assume that all the entries $RMQ[i,j,q-1]$ for $1 \leq i \leq m$ and $1 \leq j \leq n$ have been computed already. Then, $RMQ[i,j,q] = \max \set{RMQ[i-2^{q-1},j-2^{q-1},q-1], RMQ[i,j-2^{q-1},q-1], RMQ[i-2^{q-1},j,q-1], RMQ[i,j,q-1]}$. Cleary, we need $O(N)$ time to compute $RMQ[\cdot,\cdot,q]$. So, to compute the entire array $RMQ$ we need $O(N\log m)$ time. 

\paragraph{A RQM-based solution for \Problem{LCS-$\Sigma$R}.} Our second solution to \Problem{LCS-$\Sigma$R} is described in the following. 

\begin{lemma}\label{lem:lcsSRlog}
\Problem{LCS-$\Sigma$R} can be solved in $O(N\log m)$ time.
\end{lemma}
\begin{proof}
{\em Main idea.} We compute an $m \times n$ array $M[\cdot,\cdot]$ such that $M[i,j]=p$ if and only if $p$ is the largest number for which there exists a $(right)$-subsequence $s_p$ of $v[1:i]$ and $w[1:j]$ such that there are two embeddings $e_v$ and $e_w$ of $s_p$ into $v[1:i]$ and $w[1:j]$, respectively, with $e_v(p)=i,e_w(p)=j$. 

The main observation is that $M[i,j]=p$ if and only if $v[i]=w[j]=a$, where $right(a)=(\ell_a,u_a)$, and the maximum value in $M[I,J]$, for $I=[i-u_a-1:i-\ell_a-1]$ and $J=[j-u_a-1:j-\ell_a-1]$, is $p-1$. That is, there exist $i'\in I$ and $j'\in J$ such that $M[i',j']=p-1$, and this indicates the existence of a $(right)$-subsequence $s_{p-1}$ of $v[1:i']$ and $w[1:j']$, with $|s_{p-1}|=p-1$, and whose last symbol is mapped, respectively, to $v[i']$ and $w[j']$, and, moreover, this subsequence $s_{p-1}$ can be extended to a $(right)$-subsequence $s_{p}$ of $v[1:i]$ and $w[1:j]$ whose last symbol is mapped, respectively, to $v[i]$ and $w[j]$. 

{\em Preprocessing.} To begin with, we set all values of $M$ to be $0$. Then, for $j$ from $1$ to $n$, we set $M[1,j]=1$ if $v[1]=w[j]$, and for $i$ from $1$ to $m$, we set $M[i,1]=1$ if $v[i]=w[1]$. We also compute the data structure $RMQ[\cdot,\cdot,\cdot]$ for $M$. 

{\em Dynamic programming algorithm.} We compute the entries of $M[\cdot,\cdot]$ as follows:
\begin{itemize}
\item[1:] for $i=2$ to $m$ do
\item[2:] \hspace*{.5cm} for $j=2$ to $n$ do
\item[3:] \hspace*{1cm} get the constraint $right(a)=(\ell_a,u_a)$ for $a = v[i] = w[j]$; 
\item[4:] \hspace*{1cm} set $I=[i-u_a-1:i-\ell_a-1]$ and $J=[j-u_a-1:j-\ell_a-1]$; 
\item[5:] \hspace*{1cm} set $M[i,j]$ as the maximum of $M[I,J]$, computed using $RMQ$; 
\item[6:] \hspace*{1cm} set $RMQ[i,j,0]=M[i,j]$; 
\item[7:] \hspace*{1cm} for $q=1$ to $\lceil \log m \rceil$
\item[8:] \hspace*{1.5cm} Set $RMQ[i,j,q] = \max \set{RMQ[i-2^{q-1},j-2^{q-1},q-1], \\
\hspace*{2cm}  RMQ[i,j-2^{q-1},q-1], RMQ[i-2^{q-1},j,q-1], RMQ[i,j,q-1]}$
\end{itemize}

After we have computed all entries of $M$, we can simply compute its maximum in $O(N)$ time, and return it as the output value for the \Problem{LCS-$\Sigma$R}.

{\em Conclusion.} The correctness of our algorithm can be shown as follows. Clearly the entries of the arrays $M[1,\cdot]$ and $M[\cdot,1]$ are correctly computed, and the data structure $RMQ$ is correctly initialized. Assume now that $M[i',j']$ is correctly computed and the data structure $RMQ$ correctly returns $RMQ[i',j',q]$ for all $i'\leq i$ and $j'\leq j$ such that $(i',j')\neq (i,j)$. By our observations made above, the computation from step $5$ of $M[i,j]$ is therefore correct. Moreover, once $M[i,j]$ is computed, steps $6,7,8$ ensure that our $RMQ$ data structure correctly returns $RMQ[i',j',q]$ for all $i'\leq i$ and $j'\leq j$ such that $(i',j')\neq (i,j)$. So, by induction, it follows that all the entries of $M$ are correctly computed.

The overall time complexity of the algorithm is, clearly, $O(N\log m)$. 
\qed\end{proof}

\paragraph{Solutions for \Problem{LCS-$\Sigma$L} and \Problem{LCS-$\Sigma$}.} Further, we consider the problem \Problem{LCS-$\Sigma$L}, for which $right(a)=(0,n)$ for all $a\in \Sigma$. The input of this problem consists in two words $v,w$ and one function $left$ with $left(a) = (\ell_a, u_a)$ for all $a \in \Sigma$. We can immediately transform this problem into \Problem{LCS-$\Sigma$R} with input words $v^R$ and $w^R$ (i.\,e., the mirror images of the input words) and the function $right'$ which defines the gap constraints, where $right'(a)=left(a)$, for all $a\in \Sigma$. Then we can use the solutions we have seen in Lemmas \ref{lem:lcsSRsigma} and \ref{lem:lcsSRlog} to get a solution for \Problem{LCS-$\Sigma$L}.

\begin{theorem} \label{thm:lcsSR}
\Problem{LCS-$\Sigma$R}, \Problem{LCS-$\Sigma$L} can be solved in $O(\max\{N\sigma,N \log m\})$ time.\looseness=-1
\end{theorem}

Our approaches can be generalized to solve the general problem \Problem{LCS-$\Sigma$}. 

\begin{theorem}\label{lem:lcsS}
\Problem{LCS-$\Sigma$} can be solved in $O(\min\{N\sigma^2, N\sigma\log m\})$ time.
\end{theorem}
\begin{proof}
{\em Main idea}. As in the case of \Problem{LCS-$\Sigma$R}, our approach is to compute for each pair of positions $(i,j)\in [m]\times [n]$, such that $v[i]=w[j]=a$, the longest $(left,right)$-subsequence $s_p$ of both $v$ and $w$, where $|s_p|=p$ and the last symbol of $s_p$ is mapped to $v[i]$ in the embedding of $s_p$ in $v[1:i]$ and to $w[j]$ by the embedding of $s_p$ in $w[1:j]$. This can be obtained by extending with the symbol $a=v[i]=w[j]$ the longest $(left,right)$-subsequence $s_{r}$, of length $r$, whose last symbol, say $b$, is mapped to position $i'$ of $v$ and to position $j'$ of $w$, such that the gap between $i'$ and $i$ and the gap between $j'$ and $j$ fulfils the gap constraint defined by the pair $(left(b),right(a))$; clearly, $p$ is then defined as $r+1$. 

To find, for some $i$ and $j$ such that $v[i]=w[j]$, the longest $(left,right)$-subsequence $s_{r}$ which can be extended with the symbol $a=v[i]=w[j]$ to a longer $(left,right)$-subsequence we proceed as follows. Let $right(a)=(\ell_a,u_a)$. For each letter $b\in \Sigma$, let $left(b)=(\ell'_b,u'_b)$. We compute the pair $(\ell_{ab},u_{ab})=(\max\{\ell_a,\ell'_b\}, \min\{u_a,u'_b\})$ and let $I_{ab}=[i-u_{ab}-1:i-\ell_{ab}-1]$ and $J_{ab}=[j-u_{ab}-1:j-\ell_{ab}-1]$. Then, we compute $M_{ab}$ the maximum entry $M[i'][j']$ of $M[I_{ab}][J_{ab}]$ with $v[i']=w[j']=b$. Therefore, we need a mechanism allowing us to extract from $M$ the maximum from a submatrix, but only consider the entries of this submatrix that correspond to a certain letter.

{\em Dynamic Programming.} To achieve this, we compute $\sigma$ $m\times n$ matrices $M_b$, for all $b\in \Sigma$ and an $m\times n$ matrix $M$. We define $M[i,j]$ as the longest $(left,right)$-subsequence $s_p$ of both $v$ and $w$, where $|s_p|=p$ and the last symbol of $s_p$ is mapped to $v[i]$ in the embedding of $s_p$ in $v[1:i]$ and to $w[j]$ by the embedding of $s_p$ in $w[1:j]$; clearly, $M[i,j]\neq 0$ if and only if $v[i]=w[j]$. Then, for all $b\in \Sigma$, $M_b[i,j]$ is initialized as $0$ and whenever we set $M[i,j]=x$, if $v[i]=w[j]=b$ then we also set $M_b[i,j]=x$.  

Now, to compute $M[i,j]$, if $v[i]=w[j]=a$, we need to retrieve, for all $b\in \Sigma$, the largest entry $M_{ab}$ of the submatrix $M_b[I_{ab},J_{ab}]$, and set $M[i,j]=\max_{b\in \Sigma} M_{ab}+1$. 

To avoid repeating the same algorithms again, we only describe how this is done informally. On the one hand, we can use the approach from Lemma \ref{lem:lcsSRsigma}, and maintain data structures ${\mathcal D_{ab}}$ for each matrix $M_b$, for all $a,b\in \Sigma$ (so, in total, maintain $\sigma^2$ such data structures). Then, using the same algorithmic approach outlined in Lemma \ref{lem:maxdeq}, the computation of all entries $M[i,j]$ can be done in $O(N\sigma^2)$ time. On the other hand, we can use the approach from Lemma \ref{lem:lcsSRlog}, and maintain RMQ data structures for each matrix $M_b$, for all $b\in \Sigma$. Using the same algorithmic approach as in Lemma \ref{lem:lcsSRlog}, the computation of all entries $M[i,j]$ can be done in $O(N\sigma\log m)$ time. 

{\em Conclusion.} The correctness of this approach follows from the explanations given above and the discussions and proofs regarding {\Problem{LCS-$\Sigma$R}}. Its time complexity is $O(\min\{N\sigma^2, N\sigma\log m\})$.
\qed\end{proof}

%% file: bwin.tex
\section{LCS with Global Constraints}
\label{sec:bwin}
\vspace*{-0.1cm}

In this section, we present our solution to \Problem{LCS-BR}. First, we note the na\"ive solution: we consider every pair $(v[i+1:i+B],w[j+1:j+B])$ of factors of length $B$ of the two input words, respectively, and find their longest common subsequence, using the folklore dynamic programming algorithm for \Problem{LCS}. As each word of length $n$ has $n - B + 1$ factors of length $B$, this approach requires solving LCS for $O((m - B)(n-B))\subseteq O(N)$ words, with each such LCS-computation requiring $O(B^2)$ time. This yields a total time complexity of $O(N B^2)$.

Here, we improve this by providing an $O(N B^{o(1)})$ time algorithm via the \emph{alignment oracles} provided by Charalampopoulos et al. \cite{Charalampopoulos21}. Each such oracle is built for a pair of words $v$ and $w$, with $|v|=m$, $|w|=n$ and $N=mn$, and is able to return the answer to queries asking for the length of the LCS between two factors $v[i:i']$ and $w[j:j']$. One of the results of \cite{Charalampopoulos21} is the following theorem.

\begin{theorem}[\cite{Charalampopoulos21}]
    \label{thm:oracle}
    We can construct in $N^{1+o(1)}$ time an alignment oracle for the words $v$ and $w$, with $\log^{2+o(1)} N$ query time.
\end{theorem}

Theorem \ref{thm:oracle} can be used directly to build an alignment oracle $\mathcal{A}$ for the input words $w$ and $v$ in $O(N^{1 + o(1)})$ time.
Using this oracle, we can solve $\Problem{LCS-BR}$ by making $O(N)$ queries to $\mathcal{A}$, for every pair of indices $i\leq m,j \leq n$, each requiring $O(\log^{2 + o(1)} N)$ time.
The total time complexity of this direct approach is therefore $O(N^{1 + o(1)})$. We improve this approach by creating a set of smaller oracles, allowing us to avoid the extra work required to answer $\Problem{LCS}$ queries beyond the bounded range. For simplicity, assume w.l.o.g. that $m$ and $n$ are multiples of $2B$.
Let, for all $i$, $w_i = w[i B + 1: (i + 2)B]$, or $w_i=w[i B + 1: n]$ if $(i + 2) B > n$ and $v_i = v[i B + 1: (i + 2)B]$, or $v_i=v[i B + 1: n]$ if $(i + 2) B > m$.
Observe that every factor of length $B$ of $w$ appears in at least one subword $w_i$ and every factor of length $B$ of $v$ appears in at least one subword $v_i$ .
Therefore, solving $\Problem{LCS-BR}$ with input $v_i, w_j$ for every $i \in [0, m/B], j \in [0, n/B]$ would also give us the solution to $\Problem{LCS-BR}$ for input words $v,w$.

To find the solution to $\Problem{LCS-BR}$ with input $v_i, w_j$, we use Theorem \ref{thm:oracle} to construct an oracle $\mathcal{A}_{i,j}$ for $v_i$ and $w_j$.
As $|v_i| = |w_j| = 2 B$, the oracle $\mathcal{A}_{i,j}$ can be constructed in $O(B^{2 + o(1)})$ time.
The solution of $\Problem{LCS-BR}$ for the input words $v_i, w_j$ can be then determined by making $O(B^2)$ queries to $\mathcal{A}_{i,j}$, each requiring $O(\log^{2 + o(1)} B)$ time.
So, the total time complexity of solving $\Problem{LCS-BR}$ for input words $v_i, w_j$ is $O(B^{2 + o(1)})$. As there are $O(N / B^2)$ pairs of factors $v_i,w_j$, the time complexity of solving $\Problem{LCS-BR}$ for words $v,w$ is 
$O\left(\frac{N}{B^2} B^{2 + o(1)}\right) = O(NB^{o(1)}).$

\begin{theorem}
	\label{thm:lcs_b} 
	\Problem{LCS-BR} can be solved in $O(NB^{o(1)}) $time. 
\end{theorem}

 We complement our exact algorithm with an $O(N)$ time constant-factor approximation algorithm.
As before, we use the subwords $v_i = v[i B + 1 : (i + 2) B]$ and $w_j = w[j B + 1 : (j + 2) B]$.
Letting $s$ be the longest common subsequence within a bounded range of length $B$ between $v$ and $w$, note that $|s| \leq \max_{i,j \in [0, n/B]} LCS(v_i, w_j) \leq 3 |s|$.
Therefore, $\max_{i,j \in [0, n/B]} LCS(v_i, w_j) / 3$ is at most a $(1/3)$-approximation of the longest common $B$-subsequence between $v$ and $w$.

\begin{theorem}
	\label{thm:lcs_b_approx}
	Given a pair of words $v, w \in \Sigma^*$, of length at most $n$, and let $v_i = v[i B + 1: (i + 2)B]$ and, respectively, $w_j = w[j B + 1: (2 + 1)B]$.
	The length of the longest common $B$-subsequence $s$ between $v$ and $w$ has the following bound:
	$$\max_{i,j \in [n / B]} LCS(v_i,w_j)/3 \leq  |s| \leq \max_{i,j \in [n / B]} LCS(v_i,w_j)$$
    where $LCS(v_i, w_j)$ is the length of longest common subsequence between the strings $v_i$ and $w_j$.
	Further, these bounds can be found in $O(N)$ time.
\end{theorem}

\begin{proof}
	Let $x,y \in [B + 1, n]$ be the pair of indices maximising $LCS(v[x - B: x], w[y - B: y])$, and let $s$ be the longest common subsequence between $v[x - B: x]$ and $w[y - B: y]$.
	Note that $s$ must appear as a subsequence of $v_i$ and $w_j$ where $i = \lfloor \frac{x}{B} \rfloor$ and $j = \lfloor \frac{y}{B} \rfloor$.
    Therefore, as $LCS(v_i, w_j) \geq |s|$, the length of $s$ is at $\max_{i,j \in [n/B]} LCS(v_i, w_j)$.
    
    On the other direction, assume for the sake of contradiction, that there exists some $v_i, w_j$ such that $|s| < LCS(v_i, w_j) / 3$.
    Let $s'$ be the longest common subsequence between $v_i$ and $w_j$ such that $|s'| > 3 |s|$.
	Now, let $(x_1, y_1)$, $(x_2, y_2)$, $\dots$, $(x_{|s'|}, y_{|s'|})$ be a set of tuples of indices such that $v[x_{\ell}] = w[y_{\ell}] = s'[\ell]$.
 % and $x_{\ell} \in [x - B, x], y_{\ell} $.
	Observe that there exists a set of indices such that, for every $\ell \in [2, |s'|]$, $x - B \leq x_{\ell - 1} < x_{\ell} \leq x$ and $y - B \leq y_{\ell - 1} < y_{\ell} \leq y$.
	Therefore, if there exists some pair of indices $x_{\ell} > i  B + B$, $y_{\ell} \leq j B + B$, there can not exist any pair $x_{\ell'} \leq i B + B$, $y_{\ell'} > j B + B$.
	Conversely if there exists some pair $x_{\ell} \leq i B + B$, $y_{\ell'} > j B + B$, then there can not exist any pair $x_{\ell'} \leq i B + B$, $y_{\ell'} > j B + B$.
	Therefore, there must be a set of sequences $s'_1, s'_2, s'_3$ such that $s' = s'_1, s'_2, s'_3$, and:
	\begin{itemize}
		\item $s'_1$ is a subsequence of $v_i[1: B]$ and $w_j[1:B]$.
		\item $s'_2$ is a subsequence of either $v_{i}[1:B]$ and $w_{j}[B + 1: 2B]$ or $v_{i}[B + 1: 2B]$ and $w_j[1, B]$.
		\item $s'_3$ is a subsequence of $v_{i}[B + 1: 2B]$ and $w_{j}[B + 1: 2B]$.
	\end{itemize}
	As $|s'_1| + |s'_2| + |s'_3| = |s'|$, then the length of at least one of $LCS(v_i[1:B], w_j[1:B])$, $LCS(v_i[1: B], w_{j})[B + 1:2B]$, $LCS(v_{i}[B + 1:2B], w_j[1: B])$ or $LCS(v_{i}[B + 1: 2B], w_{j}[B + 1: 2B])$ must be at least $|s'| / 3 \leq |S|$.
	Hence,
	$$\max_{i,j \in [n / B]} LCS(v_i,w_j)/3 \leq  |s| \leq \max_{i,j \in [n / B]} LCS(v_i,w_j).$$
	The time complexity follows directly.
\qed\end{proof}

%% file: conclusion.tex
\section{Future Work}
\label{sec:conclusion}

A series of problems remain open from our work. Can our results regarding \Problem{LCS-MC} be improved, at least in its particular case \Problem{LCS-O(1)C}? If not, can one show tight complexity lower bounds for these problems? Can the dependency of $\Sigma$ from the solutions to \Problem{LCS-$\Sigma$} and its variants be removed? We were not focused on shaving polylog factors from the time complexity of our algorithms, but it would be also interesting to see if this is achievable. Nevertheless, it would be interesting to address also the problem of efficiently computing the actual longest common constrained subsequences in the case of all addressed problems. 
%Finally, it would be interesting to consider, following \cite{DayKMS22}, \Problem{LCS} problems for subsequences with regular-language constraints.

%\begin{itemize}
%	\item summary of all results in a table
%	\item open problem lower bounds
%	\item open problem add regular constraints
%	\item open problem retrieve the LCS by backtracking (non trivial for all subproblems)
%\end{itemize}